\documentclass[journal,draftcls,onecolumn,12pt,twoside]{IEEEtranTCOM}
\usepackage{amssymb,amsmath}
\usepackage{bbding}
\usepackage{flushend,cuted}
\usepackage{amsthm}
\usepackage{amsfonts}
\usepackage{float} % is used for \begin{figure}{H}
\usepackage{srcltx}
\usepackage{mathrsfs} % is required for \mathscr alwhichphabet.
\usepackage{multirow}
\usepackage{amssymb}

\usepackage{graphicx}
\usepackage{epsfig}
\usepackage{psfrag}
\usepackage{subfigure}

\usepackage{setspace}
\usepackage{multicol}
\usepackage[noadjust]{cite}

\usepackage[]{units} % for nicefrac
\usepackage{url} % for URL bib
\usepackage[dvips]{color}
\usepackage{verbatim} % for comment environment
\usepackage{cite} % for [1]-[3]
\usepackage{ifthen} % for TpX
\usepackage{ifpdf} % for TpX
\usepackage{soul}
\usepackage{mathtools}

\newtheorem{lemma}{Lemma}

\newtheorem{theorem}{Theorem}

\DeclareMathAlphabet{\mathbsf}{OT1}{cmss}{bx}{n}% bold sans serif
\DeclareMathAlphabet{\mathssf}{OT1}{cmss}{m}{sl}% slanted sans serif
\DeclareMathAlphabet{\mathcsf}{OT1}{cmss}{sbc}{n}% condensed sans serif

\newcommand{\svv}[1]{\mathbf{#1}}

\DeclareSymbolFont{bsfletters}{OT1}{cmss}{bx}{n}  
\DeclareSymbolFont{ssfletters}{OT1}{cmss}{m}{n}
\DeclareMathSymbol{\bsfGamma}{0}{bsfletters}{'000}
\DeclareMathSymbol{\ssfGamma}{0}{ssfletters}{'000}
\DeclareMathSymbol{\bsfDelta}{0}{bsfletters}{'001}
\DeclareMathSymbol{\ssfDelta}{0}{ssfletters}{'001}
\DeclareMathSymbol{\bsfTheta}{0}{bsfletters}{'002}
\DeclareMathSymbol{\ssfTheta}{0}{ssfletters}{'002}
\DeclareMathSymbol{\bsfLambda}{0}{bsfletters}{'003}
\DeclareMathSymbol{\ssfLambda}{0}{ssfletters}{'003}
\DeclareMathSymbol{\bsfXi}{0}{bsfletters}{'004}
\DeclareMathSymbol{\ssfXi}{0}{ssfletters}{'004}
\DeclareMathSymbol{\bsfPi}{0}{bsfletters}{'005}
\DeclareMathSymbol{\ssfPi}{0}{ssfletters}{'005}
\DeclareMathSymbol{\bsfSigma}{0}{bsfletters}{'006}
\DeclareMathSymbol{\ssfSigma}{0}{ssfletters}{'006}
\DeclareMathSymbol{\bsfUpsilon}{0}{bsfletters}{'007}
\DeclareMathSymbol{\ssfUpsilon}{0}{ssfletters}{'007}
\DeclareMathSymbol{\bsfPhi}{0}{bsfletters}{'010}
\DeclareMathSymbol{\ssfPhi}{0}{ssfletters}{'010}
\DeclareMathSymbol{\bsfPsi}{0}{bsfletters}{'011}
\DeclareMathSymbol{\ssfPsi}{0}{ssfletters}{'011}
\DeclareMathSymbol{\bsfOmega}{0}{bsfletters}{'012}
\DeclareMathSymbol{\ssfOmega}{0}{ssfletters}{'012}

\newcommand\figSize{0.75}
\DeclareRobustCommand{\PowerMat}[1][Nt]{\ensuremath{P}}

\newcommand{\SNR}{\text{$\mathsf{SNR}$}}
\DeclareRobustCommand{\prob}[1][{\rm Pr}]{\ensuremath {{#1}}}
\DeclareRobustCommand{\aNorm}[1][aNorm]{\ensuremath {\|{\bf a}\|}}
\DeclareRobustCommand{\Nt}[1][Nt]{\ensuremath {N}}
\DeclareRobustCommand{\alNt}[1][Nt]{\alpha(\Nt)}

\usepackage{lipsum}

\usepackage{fancyhdr}
\begin{document}
\allowdisplaybreaks

%\fancyfoot[LE,RO]{\thepage}           % page number in "outer" position of footer line
%\fancyfoot[RE,LO] % other info in "inner" position of footer line
\title{Simple Bounds for the Symmetric Capacity of the Rayleigh Fading Multiple Access Channel}
%\title{Distributed Expurgation for Discrete Multiple-Access Channel via Linear Codes}

%\author{Elad Domanovitz and Uri Erez, \textit{Member, IEEE}

% \author{\IEEEauthorblockN{Michael Shell}
% \IEEEauthorblockA{School of Electrical and\\
% Computer Engineering\\
% Georgia Institute of Technology\\
% Atlanta, Georgia 30332--0250\\
% Email: mshell@ece.gatech.edu}
% \and
% \IEEEauthorblockN{Homer Simpson}
% \IEEEauthorblockA{Twentieth Century Fox\\
% Email: homer@thesimpsons.com}
% \and
% \IEEEauthorblockN{James Kirk\\
% and Montgomery Scott}
% \IEEEauthorblockA{Starfleet Academy\\
% San Francisco, California 96678-2391\\
% Telephone: (800) 555--1212\\

\author{Elad Domanovitz and Uri Erez\thanks{The work of E. Domanovitz and U. Erez was supported in part by the Israel
Science Foundation under Grant No. 1956/15 and by the Heron consortium
via the Israel Ministry of Economy and Industry.

The material in this paper was presented in part at the 2018 IEEE
International Symposium on Information Theory, Vail, CO.

E. Domanovitz and U. Erez are with the Department of Electrical Engineering -–- Systems, Tel Aviv University, Tel Aviv, Israel (email: domanovi,uri@eng.tau.ac.il).}
% \and
% \IEEEauthorblockN{Uri Erez}
% \IEEEauthorblockA{%Dept. of EE-Systems,
% %TAU\\
% Tel Aviv University \\
% Email: uri@eng.tau.ac.il}

%\thanks{This work was supported in part by the Israel Science Foundation under Grant No. 1956/15.}
%\thanks{E. Domanovitz and U. Erez are with the Department of Electrical Engineering -- Systems, Tel Aviv University, Tel Aviv, Israel (email: domanovi,uri@eng.tau.ac.il).}
}
\maketitle

\begin{abstract}
%To be considered %for the 2018 IEEE %Jack Keil Wolf %ISIT Student Paper %Award.

%THIS PAPER IS ELIGIBLE FOR THE STUDENT PAPER AWARD.

Communication over the i.i.d. Rayleigh slow-fading MAC is considered, where all terminals are equipped with a single antenna. Further, a communication protocol is considered where all users transmit at (just below) the symmetric capacity (per user) of the channel, a rate which is fed back (dictated) to the users by the base station. Tight bounds are established on the distribution of the rate attained by the protocol. In particular, these bounds characterize the probability that the dominant face of the MAC capacity region contains a symmetric rate point, i.e., that the considered protocol strictly attains the sum capacity of the channel. The analysis provides a non-asymptotic counterpart to the diversity-multiplexing tradeoff of the multiple access channel.
Finally, a practical scheme based on integer-forcing and space-time precoding is shown to be an effective coding architecture for this communication scenario.

\end{abstract}

%=============================================================================

%
\section{Introduction}
%
%=============================================================================
%\blfootnote{978-1-5090-2152-9/16/\$31.00 \copyright2016 IEEE}
\label{sec:intro}

In this paper we consider communication over the slow (block) fading i.i.d. Rayleigh multiple access channel (MAC). 
For a given realization of the channel gains, the channel reduces to the classical Gaussian MAC, the capacity region of which is well known, see e.g., \cite{CoverBook}.

A basic criterion for analyzing the performance of different access methods is the gap from the sum-capacity (the maximal total rate that can be achieved by all the users). We note however, that in many cases, the rate distribution between different users is also of interest and in many applications, fairness is sought and a scheme which provides (maximal) equal rate to all users is desired. 

The maximal rate that can be achieved in a system where all users have equal rate is denoted as the symmetric capacity. In case the symmetric and the sum capacity coincide (alternatively the case where the dominant face of the MAC capacity region contains a symmetric-rate point), fairness can be achieved without sacrificing performance. As this is a very desirable working point, it is of interest to investigate what is the probability of this being the case. 

Some intuition to that question can be inferred from the diversity multiplexing tradeoff (DMT) of the i.i.d. Rayleigh fading MAC which
%i.i.d. Rayleigh MIMO-MAC 
provides an asymptotic analysis of the symmetric capacity \cite{2004:DMT_MAC}. As we show next, at high values of signal-to-noise ratio (SNR), the symmetric capacity approaches the sum capacity with high probability. 

In this paper we characterize the behaviour of the symmetric capacity for finite SNR. From this characterization, the probability of getting fairness for ``free'' for all SNRs can be easily deduced.

Another motivation for studying the symmetric capacity comes from another design criterion which is the amount of coordination needed by the protocol. High level of coordination results in high throughput loss when finite block length coding is taken into account or increased latency. As the number of users that are simultaneously transmitting increases the amount of coordination increases and thus its impact increases. This is a major issue for new applications being developed for next generation wireless networks (see, e.g.,  \cite{boccardi2014five,andrews2014will})
where supporting high number of users is required along with guaranteeing low latency. 

In theory, transmission at rates approaching the symmetric capacity requires minimal coordination; namely a single parameter, the common code rate all users should use. Nonetheless, when is comes to practical schemes that are able to approach this operating point, hitherto practical applicable transmission schemes
have relied on a much higher degree of coordination. 

Specifically, both time sharing of the points achievable via successive interference cancellation as well as rate splitting are asymmetric between the users and thus require coordination. 
Furthermore, orthogonal multiple access techniques (e.g., time or frequency division multiple access) 
also require coordination to achieve its maximal achievable symmetric-rate point which further falls short of the symmetric capacity (unless the latter coincides with the sum capacity).

% Whenever multi-user transmission is considered, an important trade-off between resource utilization (and thus improved performance) and complexity exists. Dividing the resources (either time, frequency or code) between the users results with
% a simpler receiver but suffers from inherent performance loss compared to the case where several users utilize together the same resource.

% Many systems (for example, previous generations of cellular networks), use orthogonal multiple access (OMA) methods. Such examples are TDMA (where differed users use differed time slots), FDMA (where differed users use differed frequencies) and CDMA (where differed users use differed code domain). Analysis of OMA transmission ([REFERENCES]) shows that while sum capacity can be attained, it is attained at a specific rate configuration that in general does not guarantee equal rate distribution, i.e, in many cases it is far from fairness. Further, the allocation of resources between the different users requires high level of coordination and results with performance loss and increased latency.

% The need to better utilize the resources offered by a system along with a significant increase in the number of users served by a base station node has recently led to a resurgence of interest in non-orthogonal multiple access methods (NOMA) and the study of various associated communication schemes. See, e.g., \cite{} and references therein.

The contribution of the present work is two-fold:
\begin{enumerate}
    \item Establishing bounds on the gap between the symmetric capacity and sum capacity for the Rayleigh-fading MAC.
    \item Proposing a practical scheme that is able to approach the symmetric capacity with the minimal possible degree of coordination. i.e., specification of the common per-user transmission rate.
\end{enumerate}

% It was shown ([REFERENCES]) that NOMA can achieve any point in the capacity region (and as a result it can achieve the sum capacity or the symmetric capacity), however, in all the suggested methods some level of coordination is required (either define the rate per user, define the transmission order or apply time sharing).

These two points have immediate practical implications. Specifically, we are able to characterize the performance of a protocol 
where all users transmit at a rate just below the symmetric capacity (per user) of the channel. The underlying assumption is that the latter rate is dictated to the users by the base station, utilizing a minimal amount of feedback (which does not scale with the number of users).

%in which all users transmit simultaneously at the same rate, 
%a scheme requiring very low %coordination. Therefore, we consider the performance of a simple protocol, where all users transmit at a rate just below the symmetric capacity (per user) of the channel. The underlying assumption is that the latter rate is dictated to the users by the base station, utilizing a minimal amount of feedback (which does not scale with number of users).

% when the the dominant face of the MAC capacity
% region contains the symmetric-rate point, time sharing between users is needed which requires, again, coordination.

% The performance of a simple protocol is considered, where all users transmit at a rate just below the symmetric capacity (per user) of the channel. The underlying assumption is that the latter rate is dictated to the users by the base station, utilizing a minimal amount of feedback.

% The protocol we analyze is a NOMA method which requires {\emph no coordination} (while still requiring minimal feedback from the receiver). 

%The analysis provides a non-asymptotic counterpart to the diversity-multiplexing tradeoff of the MAC and will serve to obtain tight bounds on the average throughput of the described protocol. 
%It is shown that in this case the chance of getting fairness for ``free'', i.e., the chance that the symmetric capacity and sum capacity coincide is not zero. 

Our first result is an exact characterization of the performance of the suggested communication protocol, when assuming an optimal (maximum-likelihood) receiver, for the two-user case where all nodes are equipped with a single antenna. We then extend the analysis to the scenario of an $N$-user Rayleigh-fading MAC where all nodes are equipped with a single antenna. 
For this scenario, we provide inner and outer bounds on  performance. We then further extend the analysis to a general symmetric i.i.d. Rayleigh multiple-input multiple-output (MIMO) MAC.
%and provide an achievable %bound on the performance of the protocol. 
% Since a closed-form expression can be derived also for the 2-user case where the receiver has two antennas (and each user has a single antenna) we provide the explicit characterization for this case as well. 

The derived tight characterization of the distribution of the symmetric capacity can serve as a basis for deriving other figures of merit (such as the ergodic capacity or the outage probability for any target rate). It is worthwhile noting that although the problem studied is by now quite classical,
and the derivation relies only on elementary techniques,
the obtained results--given in the from of simple closed-form expressions--appear to have eluded previous studies. 

% characterization of the performance of the suggested communication protocol is achieved using only elementary techniques and as is shown, is given by a simple closed form expression, and it appears that results have not appeared in the literature.  

Since the complexity of maximum-likelihood (ML) receiver is prohibitive, we also consider the performance attained by a practical integer-forcing (IF) receiver, demonstrating that it performs quite well in the considered scenario. Interestingly, we observe that  in order to approach the symmetric capacity with an IF receiver, another lesson from the MIMO-MAC DMT analysis should be followed. Specifically, it is necessary to apply ``space-time" precoding at the transmitters (see, e.g., \cite{2011:Hollanti_DMT_optimal_codes}).
% In order to achieve the DMT, special class of space-time precoding needs to be used. Nevertheless, we further demonstrate that even random space-time precoding can provide significant performance gain when combined with integer-forcing. 

The rest of this paper is organized as follows. Section~\ref{sec:ProblemForm} provides the problem formulation. Section \ref{sec:DMT} recounts the DMT of the i.i.d. Rayleigh-fading MIMO-MAC. As mentioned above, this asymptotic analysis provides  intuition and tools that are subsequently refined to a full characterization of the considered communication protocol. In Section~\ref{sec:sigleAnt}, the performance of the protocol is analyzed for the  case where all terminals are equipped with a single antenna. In Section~\ref{sec:Nr_Nt_MIMO_MAC}, bounds are derived for the general case of $N$ users, where each user has $N_t$ antennas and the receiver is equipped with $N_r$ antennas.
% For the two-user case, a full characterization is given while for the $N$-user case lower and upper bounds are provided. In Section \ref{sec:Nr_Nt_MIMO_MAC} an achievable bound for a general $N_r\times_N_t$ i.i.d. Rayleigh MIMO channel is derived. 
%In Section~\ref{sec:2x2_case},  a full characterization of performance is given also 
%given for the special case of a two-user i.i.d. Rayleigh-fading MIMO where each user has a single transmit antennas and the receiver is equipped with two antennas antennas, 
In Section~\ref{sec:IF}, it is demonstrated that an IF receiver combined with (structured or random) space-time precoding yields performance that is close to the established theoretical limits of the proposed communication protocol. Finally, Section~\ref{sec:conc} concludes the paper.

\section{Problem Formulation and Preliminaries}
\label{sec:ProblemForm}
% The lower bound derived in the previous section can be easily shown to cover also the case of a $2\times \Nt$ channel where $\Nt\geq2$. In this section we  adapt the bound for the case of a $1\times \Nt$ system where in this case we are interested in a MAC scenario, that is the encoders correspond to different (and distributed) users. More specifically, we analyze the ML performance of a Rayleigh fading MAC where all terminals are equipped with a single antenna and where we consider a simple transmission protocol
% as described below.

%, the receiver is also equipped with a single antenna, and a simple %transmission scheme is required such that all users transmit at the same %rate.

To simplify derivations, we will assume throughout that all users are equipped with the same number of transmit antennas. The results can  easily be extended to a more general scenario.
%where each user has $N_{t,i}$ antennas.

Accordingly, we consider a MIMO-MAC with $N$ users, where each transmitter has $N_t$ antennas and the receiver is equipped with $N_r$ antennas. The channel model can be expressed as
\begin{align}
    \svv{y}=\sum_{i=1}^{N}\svv{H}_i\svv{x}_i+\svv{n}
    \label{eq:MIMO_MAC}
\end{align}
where $\svv{H}_i$ is the channel matrix between user $i$ and the receiver. We assume an i.i.d. Rayleigh-fading model so that $\svv{H}_{i}\sim\mathcal{CN}(0,\SNR \cdot \svv{I}_{N_r})$  and $\svv{n}\sim\mathcal{CN}(0,\svv{I}_{N_r})$, where there is no statistical dependence over space nor time.\footnote{The time index remains implicit since it plays no role in the analysis. Of course, coding over large blocklength is needed to approach the information-theoretic limits.} We assume that the transmitted data $\svv{x}_i\in\mathbb{C}^{N_t\times1}$ is isotropic (``white'') for each user and that all users are subject to the same power constraint $P$
%=\svv{I}_{N_t}$, 
where the SNR is absorbed in the channel gains.

Define a subset of users by $\mathcal{S} \subseteq \{1, 2,\ldots, N\}$. Then,  the capacity region of the channel is given by (see, e.g., \cite{CoverBook}) all rate vectors $(R_1, \ldots, R_{N})$ satisfying
\begin{align}
\sum_{i\in \mathcal{S}} R_i & \leq C(\mathcal{S}) \nonumber \\
& \triangleq \log\det\left(\svv{I}_{N_r}+\sum_{i\in \mathcal{S}}\svv{H}_i\svv{H}^H_i\right),
\label{eq:capacity_region_MAC}
\end{align}
for all subsets $\mathcal{S}$ in the power set
of $\{1, 2,\ldots, N\}$. The sum capacity is given by 
\begin{align}
C& \triangleq C(\{1, 2,\ldots, N\}) \\
& = \log\det\left(\svv{I}_{N_r}+\sum_{i=1}^N\svv{H}_i\svv{H}^H_i\right).
\label{eq:sum_capacity}
\end{align}

If we impose the constraint that all users transmit at the same rate, then the maximal achievable rate is given by substituting $R_i=C_{\Sigma-\rm sym}/N$ in \eqref{eq:capacity_region_MAC}, from which it follows that the symmetric capacity $C_{\Sigma-\rm sym}$ is dictated by the bottleneck:
\begin{align}
C_{\Sigma-\rm sym}
%&=\min_{S\subseteq\{1, 2,\ldots,\Nt\}}\frac{\Nt}{|S|}\log\det\left(\svv{I}+\PowerMat\sum_{i\in S}\svv{H}_i^H\svv{H}_i\right)
%\label{eq:sym_capacity_MAC_ver1}
%\\
&=\min_{\mathcal{S}\subseteq\{1, 2,\ldots,\Nt\}}\frac{\Nt}{|\mathcal{S}|}\log\det\left(\svv{I}+\sum_{i\in \mathcal{S}}\svv{H}_i\svv{H}_i^H\right).
\label{eq:sym_capacity_MAC_ver2}
\end{align}

% Note that \eqref{eq:sym_capacity_MAC} is a special case of (\ref{eq:R_ML_spaceOnly}).

%In a MAC channel, the  encoders are distributed and hence applying %precoding over the antennas is precluded.

% While in general applying a CUE precoding transformation ${\svv{P}}_c$ implies joint processing at the encoders, which is precluded in a MAC setting, in an i.i.d. Rayleigh fading environment, this random transformation is actually performed by nature.\footnote{This follows since the left and right singular vector matrices of the an i.i.d. Gaussian matrix $\svv{H}_c$ are equal to the eigenvector matrices of the  Wishart ensembles $\svv{H}_c\svv{H}_c^{H}$ and $\svv{H}^H_c\svv{H}_c$, respectively. The latter are known to be CUE (Haar) distributed. See, e.g., Chapter~4.6 in \cite{edelman2005random}.} Hence
% the results developed in the previous sections readily apply to this scenario.

We  study the conditional ``cumulative distribution function'':\footnote{We use quotation marks since we impose strict inequality in $C_{\Sigma-\rm sym}<R$.} \begin{align}
    \prob(C_{\Sigma-\rm sym}<R| C).
    \label{eq:cum}
\end{align}
%for an i.i.d. %Rayleigh fading %model. 
The latter quantity provides a full statistical
characterization of the performance of the transmission protocol
considered.
%where all users transmit at a rate %just below the equal-rate capacity %(per user) of the channel, where the %underlying assumption is that this %rate is dictated to the users by the %base station.
Another interpretation of \eqref{eq:cum} is as a conditional outage probability in an open-loop scenario; that is, in a scenario where
all users (when they are active) transmit at a common target rate $R$. For a given number of active users $N$,  the outage probability is then given by $\mathbb{E}[ \prob(C_{\Sigma-\rm sym}<N \cdot R| C)]$ where the expectation is over $C$ and is computed w.r.t. an i.i.d. Rayleigh distribution.

\section{Lessons from the DMT}
\label{sec:DMT}
Some insight into the performance of the considered protocol may be obtained by considering the %diversity-multipl%exing tradeoff %(DMT) of the 
DMT of the symmetric Rayleigh-fading MIMO-MAC channel, which was studied in \cite{2004:DMT_MAC}. As a special case, the scenario where all users transmit at the same rate was considered in detail, for which  a simple expression for the DMT was derived.

Specifically, the DMT of the Rayleigh MIMO-MAC with $N$ users, where each transmitter has $N_t$ antennas and the receiver has $N_r$  antennas, and where all users transmit at the same rate, is given by
\begin{align}
    d_{\rm sym}^*(r)=\begin{cases}
    d^*_{N_t,N_r}(r),~~~~~~~~r\leq\min(N_t,\frac{N_r}{N+1}) \\
    d^*_{N\cdot N_t,N_r}(N\cdot r),~r\geq\min(N_t,\frac{N_r}{N+1})
    \end{cases}
    \label{eq:DMT_symCap}
\end{align}
where $d^*_{N_t,N_r}(r)$ 
is the DMT of the i.i.d. single-user Rayleigh-fading MIMO channel with $N_t$ transmit antennas and $N_r$ receive antennas (provided that the block length $l\geq N_t+N_r+1$); see, e.g., \cite{tse}).
The function $d^*_{N_t,N_r}(r)$ is a piecewise linear curve such that $d^*_{N_t,N_r}(r)=(N_t-r)(N_r-r)$ for every integer \mbox{$r\leq\min(N_t,N_r)$}. 
% An example of the DMT for the i.i.d. Rayleigh-fading two-user  MAC where each user has a single antenna ($N_t=1$) and the receiver has two antennas ($N_r=2)$ is depicted in~Figure~\ref{fig:DMTof2users}.

% \begin{figure}[htbp]
% \begin{center}
% \includegraphics[width=\figSize\columnwidth]{DMT_MAC_2x1_2_users_new.eps}
% \end{center}
% \caption{DMT curve for a two-user Rayleigh-fading MIMO-MAC where each user has a single  antenna and the receiver is equipped with two antennas.}
% \label{fig:DMTof2users}
% \end{figure}

%
%, characterizing the high SNR %behavior of the channel.

Although the DMT analysis is asymptotic in nature, instructive lessons may nonetheless be drawn from it.
First, it is clear that in the limit of high SNR, the ratio of the symmetric capacity and  sum capacity approaches one in probability (since the DMT is strictly positive for any multiplexing gain smaller than the maximal attainable degrees of freedom).

%can be achieved by the symmetric-rate capacity (hence the symmetric-rate %capacity converges to the sum capacity at asymptotically high SNR).
More importantly, the analysis of the typical error events in the Rayleigh-fading MAC (with equal-rate transmission) reveals that with high probability, outage occurs either as if all users were considered as a single one (``antenna pooling'') or as a result of a single-user constraint constituting the bottleneck \cite{2004:DMT_MAC}. These two regimes are reflected in the two cases appearing in \eqref{eq:DMT_symCap}.

Further, it can be easily shown that for a scalar MAC ($N_r=N_t=1$) with two or more users, the antenna polling bottleneck amounts to the probability that the sum-capacity is below the target rate. As for a (symmetric) transmission protocol where the target rate is set to just below  the sum capacity,  the latter type of outage event cannot occur, it follows that the diversity gain at the maximal multiplexing gain (the maximal attainable degrees of freedom) is \emph{strictly positive}. This in turn implies that the ratio between the symmetric capacity and the sum capacity will approach $1$ quite fast as the SNR grows. 

In fact, in the case of two users, the DMT of which is depicted in Figure~\ref{fig:DMTof3users}, we show that perfect fairness may be gained ``for free'' with high probability. This is, the probability that the symmetric capacity is equal to the sum capacity approaches $1$ rather fast as a function of the SNR;
% The analysis to follow, characterizing the distribution of the rate attained by the considered transmission protocol, reveals that the probability of the symmetric rate achieving the sum-capacity is strictly positive even for practical values of SNR; 
hence, validating the intuition gained from the DMT.

\begin{figure}[htbp]
\begin{center}
\includegraphics[width=\figSize\columnwidth]{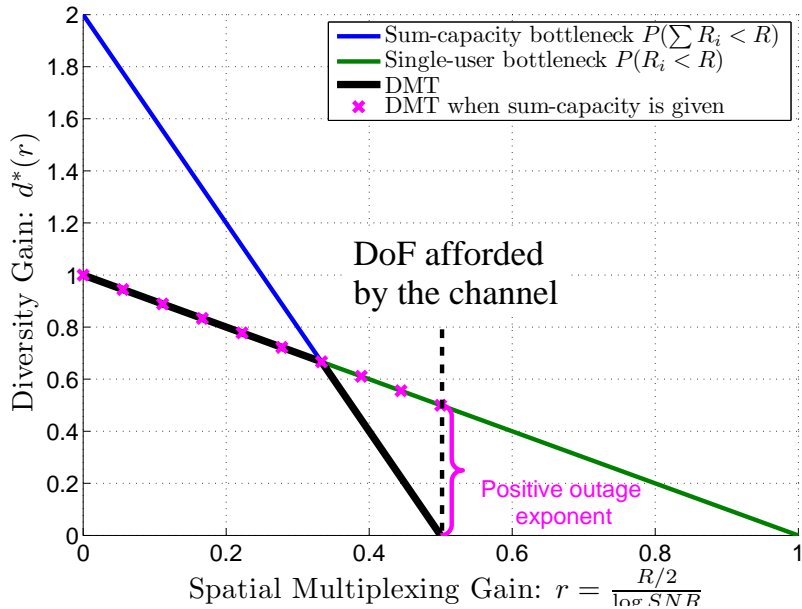}
\end{center}
\caption{DMT curve for a two-user scalar Rayleigh-fading MIMO-MAC where all terminals are equipped with a single antenna.
%each user has a single transmit antenna and the receiver is also equipped with a single antenna.
}
\label{fig:DMTof3users}
\end{figure}

\section{I.I.D. Rayleigh-fading MAC with Single-Antenna Terminals}
\label{sec:sigleAnt}

When all terminals are equipped with  a single antenna, the Rayleigh-fading MAC is described by
\begin{align}
    y=\sum_{i=1}^{\Nt}h_i x_i+n
\end{align}
% where ${h}_{i}\sim    \mathcal{CN}(0,\SNR)$  and $n\sim\mathcal{CN}(0,1)$, and where there is no statistical dependence between any of these random variables. Without loss of generality we assume throughout the analysis to follow that the signal $x_i$ of each user satisfies a  power constraint of $1$, i.e., the SNR is absorbed in the channel gains.
% The capacity region of the channel is given by (see, e.g., \cite{CoverBook}) all rate vectors
% $(R_1, \ldots, R_{N})$ satisfying the constraints
% \begin{align}
% \sum_{i\in S} R_i < \log\left(1+\sum_{i\in S}|{ h}_i|^2\right),
% %\label{eq:capacity_region_MAC}
% \end{align}
% for all ${S\subseteq\{1, 2,\ldots,\Nt\}}$.
% We denote the sum capacity by
% \begin{align}
% C=\log\left(1+\sum_{i=1}^{N}|{ h}_i|^2\right).
% \label{eq:sum_capacity_MAC}
% \end{align}
and the symmetric capacity is given by
\begin{align}
C_{\Sigma-\rm sym}=\min_{\mathcal{S}\subseteq\{1, 2,\ldots,\Nt\}}\frac{\Nt}{|\mathcal{S}|}\log\left(1+\sum_{i\in \mathcal{S}}|{ h}_i|^2\right).
\label{eq:sym_capacity_MAC}
\end{align}
\subsection{Two-user i.i.d. single antenna Rayleigh-fading MAC }
We begin by analyzing the simplest case of a two-user scalar MAC, for which we obtain an exact of  characterization of  \eqref{eq:cum}.
%$\Nt=2$.
\begin{theorem}
For a two-user i.i.d. Rayleigh-fading MAC with sum capacity $C$, for any rate $R\leq C$,
    \begin{align}
        \prob(C_{\Sigma-\rm sym}<R| C) = 2\cdot\frac{2^{R/2}-1}{2^{C}-1}.
    \end{align}
    \label{thm:thm1}
\end{theorem}
\begin{proof}
Given $C$, $\svv{h} \triangleq (h_1,h_2)$ is uniformly distributed over a two-dimensional  complex sphere of radius
$\sqrt{2^ {C}-1}$. Hence, $\svv{h}/ \|\svv{h} \|$ can be viewed as the first row of a random (Haar) unitary matrix $\svv{U}$.

%We follow the footsteps of the proof of Theorem~\ref{thm:thm1} while noting that in the case of $1\times 2 $ MAC, we have $\rho_1=2^C-1$ and $\rho_2=0$.
% By \eqref{eq:sym_capacity_MAC} and using the notation of
% \eqref{eq:joint,st}, we obtain (cf. \eqref{eq:Rj_k123})
By \eqref{eq:sym_capacity_MAC}, we obtain
\begin{align}
    C_{\Sigma-\rm sym}=\min\left\{2C(\{1\}),2C(\{2\}),C_{}\right\}.
\end{align}
We start by analyzing $2C(\{1\})$, which is given by
\begin{align}
2C(\{1\})&=2\log\left(1+|h_1|^2\right) \nonumber \\
& = 2\log\left(1+|\svv{U}_{1,1}|^2(2^{C}-1)\right).
\end{align}
% \begin{align}
%      C(\{1\})&=2\log\left(1+\begin{bmatrix}\svv{U}_{1,1}\\\svv{U}_{1,2}\end{bmatrix}^H\begin{bmatrix}
%  2^C-1 & {0} \\ {0} & {0}
%  \end{bmatrix}\begin{bmatrix}\svv{U}_{1,1}\\\svv{U}_{1,2}\end{bmatrix}\right) \nonumber \\
%  &=2\log\left(1+|\svv{U}_{1,1}|^2(2^C-1)\right).
% \end{align}
It follows that
\begin{align}
    \prob(2C(\{1\})<R|C) &=\prob\left(|\svv{U}_{1,1}|^2<\frac{2^{R/2}-1}{2^C-1}\right)
\nonumber \\ &= \prob\left(|\svv{U}_{1,1}|^2 \in \left[0,\frac{2^{R/2}-1}{2^C-1}\right)\right)
\label{eq:10}
\end{align}
Since (see, e.g., \cite{NarulaTrottWornell:1999}) for a  $2\times2$  matrix drawn uniformly with respect to the Haar measure, we have \mbox{$|\svv{U}_{1,1}|^2\sim \rm{Unif}([0,1])$}, it follows that 
\begin{align}
    \prob(2C(\{1\})<R|C)=\frac{2^{R/2}-1}{2^C-1}.
\end{align}

Now, since $\svv{U}_{1,1}$ and $\svv{U}_{1,2}$ are the elements of a row in a unitary matrix, we have
\begin{align}
    |\svv{U}_{1,1}|^2+|\svv{U}_{1,2}|^2=1.
\end{align}
Hence,
\begin{align}
    \prob(2C(\{2\})<R|C) &= \prob\left(|\svv{U}_{1,2}|^2<\frac{2^{R/2}-1}{2^C-1}\right) \nonumber \\
    & = \prob\left(1-|\svv{U}_{1,1}|^2<\frac{2^{R/2}-1}{2^C-1}\right) \nonumber \\
    & = \prob\left(|\svv{U}_{1,1}|^2 \in \left(1-\frac{2^{R/2}-1}{2^C-1} ,1\right]\right)
   %  & = \prob\left(|\svv{U}_{1,1}|^2>\frac{2^C-2^{R/2}}{2^C-1}\right)
%      \nonumber \\
%     & =1-\frac{2^C-2^{R/2}}{2^C-1} \nonumber \\
%     & = \frac{2^{R/2}-1}{2^C-1}.
\label{eq:13}
\end{align}

Since for any rate $R\leq C$, the intervals appearing in \eqref{eq:10} and \eqref{eq:13} are disjoint and of the same length, it follows that
% As shown in the Appendix, the events $\{C(\{1\})<R\}$
% and $\{C(\{2\})<R\}$ are disjoint. Since by symmetry we
% further have that
% \begin{align}
%     \prob(C(\{1\})<R|C)=\prob(C(\{2\})<R|C),
% \end{align}
% it follows that
%. Since we are interested in $0\leq R \leq C$ we conclude that
\begin{align}
    \prob(C_{\Sigma-\rm sym}<R| C) = 2\cdot\frac{2^{R/2}-1}{2^{C}-1}.
    \label{eq:nonzeroprob}
\end{align}
\end{proof}

We note that the probability in \eqref{eq:nonzeroprob}
is strictly smaller than $1$ at $R=C$. Thus, the probability that the symmetric  capacity
coincides with the sum capacity is strictly positive.
%and approaches $1$ as the capacity %grows.

Figure~\ref{fig:capRegion2users} depicts the capacity region for three different channel realizations for  which the sum capacity equals $2$. The probability that the symmetric  capacity
coincides with the sum capacity amounts to the probability that the symmetric rate line
passes through the dominant face of the capacity region and is given by
\begin{align}
    \prob\left({C_{\Sigma-\rm sym}=C|C}\right)&=1-\prob\left({C_{\Sigma-\rm sym}<C|C}\right) \nonumber \\
    & = 1 - 2\cdot\frac{2^{C/2}-1}{2^{C}-1}.
    \label{eq:green}
\end{align}
%The probability that the capacity %region is of the type of the dashed
%line (equal-rate line passing through %the dominant face of the region)
%is given by \eqref{eq:green}
%The equal rate capacity is highlighted as well and we can see that there %are cases where the equal rate capacity coincide with the sum capacity.
As an example, for $C=2$, this probability is $1/3$.
% \begin{align}
%     \prob\left({C_{\Sigma-\rm sym}=C|C=2}\right)
%     %&=1-\prob\left({C_{\Sigma-\rm sym}<C|C}\right) \nonumber \\
%     & = 1 - 2\cdot\frac{2^1-1}{2^2-1} \nonumber \\
%     & = 1 - 2\cdot\frac{1}{3} \nonumber \\
%     &=\frac{1}{3}.
%     \label{eq:eq}
% \end{align}

Figure~\ref{fig:pdf_c2} depicts
the probability density function of the symmetric  capacity
of a two-user i.i.d. Rayleigh-fading MAC given that the sum capacity is $C=2$. The probability in \eqref{eq:green} manifests itself as a delta function at the sum capacity.

% Finally,
% Figure~\ref{fig:prob_Csym_Eq_Csum} depicts the probability $\prob\left({C_{\Sigma-\rm sym}=C|C}\right)$ as a function of the sum capacity. Note that the latter probability tends to one exponentially
% fast in $C$.

% \begin{figure}
% \begin{center}
% \includegraphics[width=1\columnwidth]{figures/prob_Csym_Eq_Csum.eps}
% \end{center}
% \caption{Probability that the symmetric-rate capacity
% coincides with the sum capacity. }
% \label{fig:prob_Csym_Eq_Csum}
% \end{figure}

% by Theorem~\ref{thm:thm1} the probability that
%  the symmetric-rate equals
% achieving the white-input mutual information $C$ when using symmetric rate transmission scheme is not zero. The while-input mutual information can be viewed as the ``sum capacity'' since it suggests that the bottleneck in the capacity region  \eqref{eq:capacity_region_MAC} in the case where $\svv{S}=\{1,2,\ldots,\Nt\}$.

\begin{figure}[htbp]
\begin{center}
\includegraphics[width=\figSize\columnwidth]{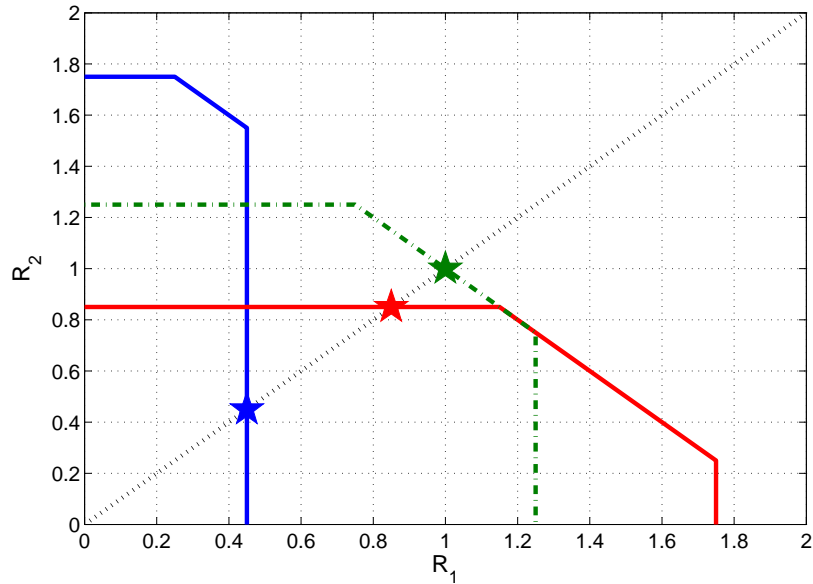}
\end{center}
\caption{Different capacity regions corresponding to a two-user  MAC with sum capacity $C=2$. For the channel depicted with a dashed-dotted line, the dominant face constitutes the bottleneck and $C_{\Sigma-\rm sym}=C$. }
\label{fig:capRegion2users}
\end{figure}

\begin{figure}[htbp]
\begin{center}
\includegraphics[width=\figSize\columnwidth]{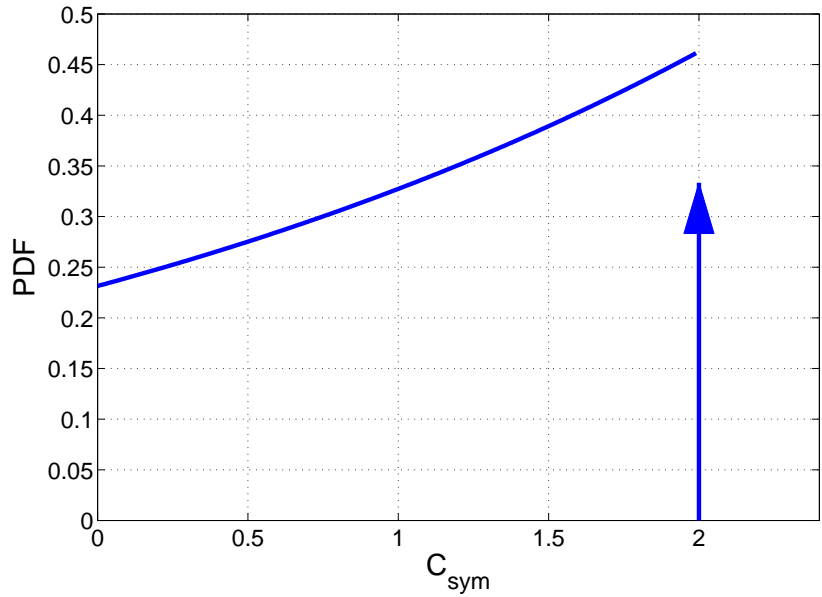}
\end{center}
\caption{Probability density function of the symmetric capacity
of a two-user i.i.d. Rayleigh-fading MAC given that the sum capacity is $C=2$.}
\label{fig:pdf_c2}
\end{figure}

\subsection{Extension to the $\Nt$-user i.i.d. scalar  Rayleigh-fading MAC}
Theorem~\ref{thm:thm1} may be extended to the case of $\Nt>2$ users. However,
rather than obtaining an exact characterization of the distribution of the symmetric  capacity, we will now be content with deriving lower and upper bounds for it.

Let us define
\begin{align}
    P_{\rm out}(k,R_{}|C)\triangleq \prob\left(\frac{\Nt}{k}C(\mathcal{S})<R_{} \Big|C_{}\right) 
    \label{def:pkout}
\end{align}
We begin with the following lemma which is the key technical step from which Theorem~\ref{thm:thm2} follows.
\begin{lemma}
For an $N$-user i.i.d. Rayleigh-fading  MAC with sum capacity $C$,
%the outage probability
and for any  subset of users $\mathcal{S}\subseteq\{1, 2,\ldots,\Nt\}$ with cardinality $k$, we have
 \begin{align*}
    P_{\rm out}(k,R_{}|C)
     &= \frac{\mathcal{B}(\frac{2^{R_{}|\mathcal{S}|/\Nt}-1}{2^{C_{}}-1};|\mathcal{S}|,\Nt-|\mathcal{S}|)}{\mathcal{B}(1;|\mathcal{S}|,\Nt-|\mathcal{S}|)}\
%\label{eq:pkout}
 \end{align*}
 where $0\leq R_{} \leq C_{}$ and
 \begin{align*}
     \mathcal{B}(x;a,b)=\int_0^{x}u^{a-1}(1-u)^{b-1}du
  %   \label{eq:incompleteBetaFunc}
 \end{align*} is the incomplete beta function.
 \label{lem:lem1}
\end{lemma}
\begin{proof}
Similar to the case of two users,  $\svv{h}\triangleq(h_1,\ldots,h_{\Nt})$ is uniformly distributed over
an $N$-dimensional complex sphere of radius $\sqrt{2^ {C}-1}$ and hence $\svv{h}/\| \svv{h} \|$ may be viewed as the first row of a unitary matrix $\svv{U}$ drawn at random according to the Haar measure.

By symmetry, for any set $\mathcal{S}$ with cardinally $k$, the distribution of
$ C(\mathcal{S})$ is equal to that of
%(and singe $\svv{U}$ is taken from the CUE, without loss of generality %taking the first $k$ channels) we have
\begin{align}
    C(\{1,2,\ldots,k\})
    &=\log\left(1+\sum_{i=1}^k |h_i|^2\right) \nonumber \\
                  &=\log\left(1+\left(2^C-1\right)\sum_{i=1}^k |U_{1,i}|^2\right).
  \label{eq:58}
\end{align}
% Denoting the size of $S$ as $|S|=k$ we note that taking $k$ columns from the equivalent channel is equivalent to applying precoding matrix composed from $k$ columns taken from the original precoding matrix.
% Denote
% \begin{align}
%     \svv{U}(k)=\begin{bmatrix}
%     U_{1,1} & \cdots & U_{1,k} \\
%     U_{2,1} & \cdots & U_{2,k} \\
%     \vdots & \vdots & \vdots \\
%     U_{\Nt,1} & \cdots & U_{\Nt,k}
%     \end{bmatrix}
% \end{align}

% we have
% \begin{align}
%     &C(\{S\}) \nonumber \\
%     &=\frac{k}{\Nt}\log\det\left(\svv{I}_{k\times k}+\svv{U}(k)^T\begin{bmatrix}
%     2^C-1 & \cdots & 0\\
%     0 & \cdots & 0 \\
%     \vdots & \vdots & \vdots \\
%     0 & \cdots & 0
%     \end{bmatrix}
%     \svv{U}(k)\right) \nonumber \\
%     &=\frac{k}{\Nt}\log\left(1+(2^C-1)\svv{U}(k)_{1}^H\svv{U}(k)_{1}\right),
% \end{align}
%where $\svv{U}(k)_{1}$ is the first row of $\svv{U}(k)$. Following \cite{dar2013jacobi} let $\lambda$ be the eigenvalue of $\svv{U}(k)_{1}^H\svv{U}(k)^H_{1}$. We therefore have
Denoting the partial sum of $k$ entries
as $\displaystyle{X=\sum_{i=1}^k |\svv{U}_{1,i}|^2}$, we therefore have
\begin{align}
   \prob\left(\frac{\Nt}{k}C(\mathcal{S})<R \Big{|}  C\right)&=\prob\left( 1+ \left(2^C-1\right)X <2^{R\frac{k}{\Nt}}\right)
   \nonumber \\
   & = \prob\left(X<\frac{2^{R\frac{k}{\Nt}}-1}{2^C-1}\right).
\end{align}
We note that the vector
$( |\svv{U}_{1,1}|^2,\ldots,|\svv{U}_{1,N}|^2 )$ follows the Dirichlet distribution and a partial sum of its entries has a Jacobi (also referred to as MANOVA) distribution.
To see this, we note that \eqref{eq:58} can be written as
\begin{align}
    \frac{\Nt}{k}C(\{1,2,\ldots,k\})=\frac{\Nt}{k}\log\left(1+(2^C-1)\svv{U}(k)_{1}\svv{U}(k)^H_{1}\right)
 \end{align}
where $\svv{U}(k)_{1}$ is a vector which contains the first $k$ elements of the first row of $\svv{U}$. Noting that since $\svv{U}(k)_{1}$ is a submatrix of a unitary matrix, its singular values follow (see, e.g., \cite{dar2013jacobi}) the Jacobi distribution, and more specifically,  $X$ has Jacobi distribution with rank 1.
% To see this we note that \eqref{eq:58} can be written as
% \begin{align}
%     C(\{1,2,\ldots,k\})=\frac{\Nt}{k}\log\left(1+(2^C-1)\svv{U}^k_{1}^H \svv{U}^k_{1}\right)
%  \end{align}
We thus obtain
%(see \cite{dar2013jacobi}),
\begin{align}
     \prob\left(\frac{\Nt}{k}C(\mathcal{S})<R_{} \Big|C_{}\right) &= \int_{0}^{\frac{2^{R_{}k/\Nt}-1}{2^{C_{}}-1}}x^{k-1}x^{\Nt-k-1}d\lambda \nonumber \\ &=\frac{\mathcal{B}\left(\frac{2^{R_{}k/\Nt}-1}{2^{C_{}}-1};k,\Nt-k\right)}{\mathcal{B}(1;k,\Nt-k)},\nonumber
 \end{align}
where $\mathcal{B}(x;a,b)$ is the incomplete beta function defined. 
%in \eqref{eq:incompleteBetaFunc}.
% \begin{align}
%      \mathcal{B}(x;a,b)=\int_0^{x}u^{a-1}(1-u)^{b-1}du \nonumber
%  \end{align}
%is the incomplete beta function.

% We note that $\svv{U}(k)=\begin{bmatrix} \svv{U}_{1,1} & \cdots & \svv{U}_{1,k}\end{bmatrix}$ is a submatrix of a unitary matrix. Let $\lambda$ be the eigenvalue of $\svv{U}(k)\svv{U}(k)^H$. We therefore have
% \begin{align}
%   \prob(C(\{S\})<R|C)&=\prob(\left(1+(2^C-1)\lambda\right)<2^{R\frac{\Nt}{k}})\nonumber \\
%   & = \prob(\lambda<\frac{2^{R\frac{\Nt}{k}}-1}{2^C-1}).
% \end{align}
% Following \cite{dar2013jacobi} we note that ${\lambda}$ has a (simple) Jacobi distribution. Therfore we get
% \begin{align}
%      \prob\left(C(\{S\})<R_{}|C_{}\right) &= \int_{0}^{\frac{2^{R_{}k/\Nt}-1}{2^{C_{}}-1}}\lambda^{k-1}\lambda^{\Nt-k-1}d\lambda \nonumber \\ &=\frac{\mathcal{B}(\frac{2^{R_{}k/\Nt}-1}{2^{C_{}}-1};k,\Nt-k)}{\mathcal{B}(1;k,\Nt-k)}\nonumber
%  \end{align}
\end{proof}
% A more explicit deviation for this Lemma appear in Appendix~\ref{sec:apcc} (where the usage of Jacobi distribution is explained through matric representation of the MAC).

%An important property of Lemma~\ref{thm:thm2} is that the outage %probability of a specific set $S$,
%As the depends only on the cardinality of the set. Hence,
%$We obtain the following.

\begin{theorem}
For an $N$-user scalar i.i.d. Rayleigh-fading MAC, we have
{
\begin{align}
 \max_k P_{\rm out}(k,R_{}|C) & \leq  \prob\left(C_{\Sigma-\rm sym} <R_{}|C\right) \\
 & \leq  \sum_{k=1}^{\Nt} {{\Nt}\choose{k}} P_{\rm out}(k,R_{}|C), \nonumber
 \end{align}
 }
where $P_{\rm out}(k,R_{}|C)$ is defined in 
\eqref{def:pkout} and given in
Lemma~\ref{lem:lem1}.
%\begin{align}
% P_{\rm out}(k,R_{}|C)\triangleq \prob(C(\{1,2,\ldots,k\})<R_{}|C_{}).
%\end{align}
\label{thm:thm2}
\end{theorem}
\begin{proof}
To establish the left hand side of the theorem, first note
that $C_{\Sigma-\rm sym}\leq C(\mathcal{S})$ for any subset $\mathcal{S}$ and hence
%Thus,
%\begin{align}
%    \prob\left(C(\{S\})<R_{}|C_{}\right)=\prob(C(\{1,2,\ldots,k\})<R_{}|C%_{})
%\end{align}
%
%As noted in the proof of Lemma~\ref{lem:lem1},
\begin{align}
C_{\Sigma-\rm sym} \leq \min_k \frac{\Nt}{k} C(\{1,2,\ldots,k\}).
\end{align}
It follows that
\begin{align}
 &\prob\left(C_{\Sigma-\rm sym} <R_{}\Big|C\right)  \nonumber \\
 &\geq \prob\left( \min_k \frac{\Nt}{k}C(\{1,2,\ldots,k\})<R_{} \Big| C\right)   \nonumber \\
 &= \prob\left( \bigcup_k \left\{ \frac{\Nt}{k}C(\{1,2,\ldots,k\})<R_{} \right\}\Big| C\right) \nonumber \\
 &\geq \max_k \prob\left(  \frac{\Nt}{k}C(\{1,2,\ldots,k\})<R_{}  \Big|  C\right) \nonumber \\
 & = \max_k P_{\rm out}(k,R_{}|C).
\end{align}
%Therefore, the left hand side of the theorem
%can be easily p by noting that $C_{\Sigma-\rm sym}$ is acquired by taking the %minimum over all values of $k$. Hence,
%follows by noticing  that
%$C_{\Sigma-\rm sym}\leq C(\{S\})$ for any subset $S$ and hence
%it follows that
%\begin{align}
%     P_{\rm out}(k,R_{}|C)~\leq~\prob\left(C_{\Sigma-\rm sym}<R_{}|C\right).
%\end{align}
%Taking the maximum over $k$ yields the left hand side.
The right hand side follows by the union bound.
%and recalling that subsets with the same cardinality have the same outage probability. %Thus we can upper bound the outage probability by multiplying the outage %probability of a specific cardinality with the number of sets which has %the same cardinality. Hence we get
%\begin{align}
%\prob\left(C_{\Sigma-\rm sym}<R_{}|C\right) & \leq \sum_{S\subseteq\{1, %2,\ldots,\Nt\}}\prob\left(C\{S\}<R_{}|C\right) \nonumber \\
%& = \sum_{k=1}^{\Nt} {{\Nt}\choose{k}} P_{\rm out}(k,R_{}|C)
%\end{align}
\end{proof}

Figures~\ref{fig:4x4_C8_perStream} and~~\ref{fig:lowUpBoundOutageMAC}
illustrate the theorem for the case of four users (where the markers indicate the height of the delta functions).
As can be seen from Figure~\ref{fig:4x4_C8_perStream}, already at not very high values of capacity, the single-user constraints already constitute the bottleneck.
We further observe from Figure~\ref{fig:lowUpBoundOutageMAC}
that the union bound is quite tight.

\begin{figure}[htbp]
\begin{center}
\includegraphics[width=\figSize\columnwidth]{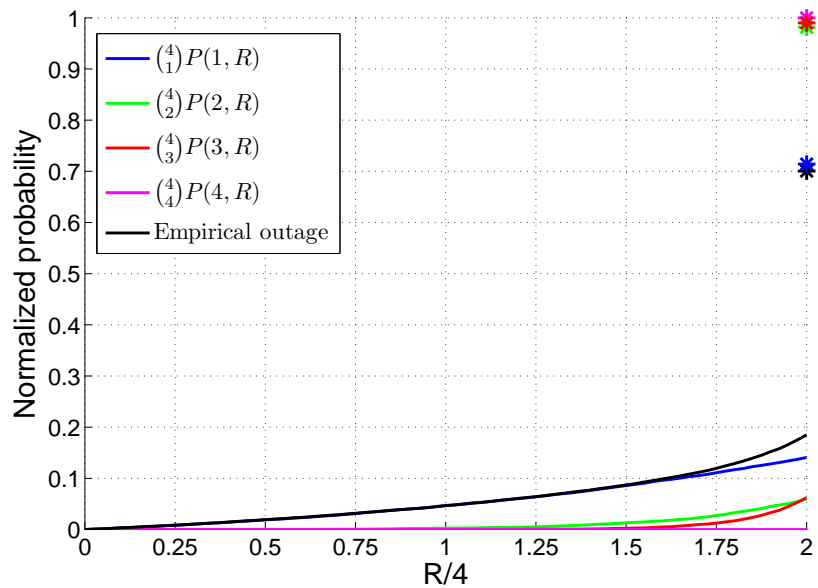}
\end{center}
\caption{Demonstration of the quantities appearing in the bounds appearing in Theorem~\ref{thm:thm2} for the case of a $4$-user i.i.d. Rayleigh-fading  MAC with sum capacity $C=8$.}
\label{fig:4x4_C8_perStream}
\end{figure}

\begin{figure}[htbp]
\begin{center}
\includegraphics[width=\figSize\columnwidth]{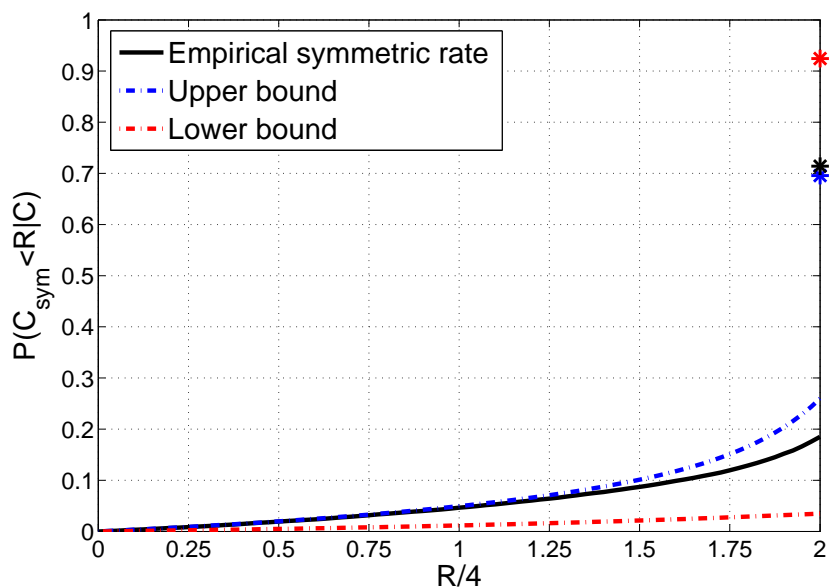}
\end{center}
\caption{Comparison of empirical evaluation of \eqref{eq:sym_capacity_MAC} and Theorem 2 (upper  and lower bounds for the outage probability) for a $4$-user i.i.d. Rayleigh-fading MAC ($N_t=N_r=1$) with sum capacity $C_{}=8$.}
\label{fig:lowUpBoundOutageMAC}
\end{figure}

\section{Upper Bound on the Outage Probability for the symmetric $N$-user  Rayleigh-fading MIMO  MAC}
%^each user with $N_{t}$ antennas and receiver with $N_r$ antennas}
\label{sec:Nr_Nt_MIMO_MAC}

We now consider the symmetric MIMO-MAC scenario where each of the 
%more general case where we have $N$-users, each 
$N$ users is equipped with $N_{t}$ antennas and the receiver is equipped with $N_r$ antennas. In this case, the channel as described by \eqref{eq:MIMO_MAC} can be rewritten as
\begin{align}
    \svv{y}=\mathcal{H}\mathcal{X}+\svv{n}
\end{align}
where 
\begin{align*}
\mathcal{H}=\begin{bmatrix} \svv{H}_1 & \svv{H}_2 & \hdots & \svv{H}_N \end{bmatrix}
%\label{eq:bigH}
\end{align*}
and 
$$\mathcal{X}=\begin{bmatrix}\svv{x}_1^T & \svv{x}_2^T   \hdots  \svv{x}_N^T \end{bmatrix}^T.$$

Therefore, the symmetric capacity \eqref{eq:sym_capacity_MAC_ver2} can be expressed as
\begin{align}
C_{\Sigma-\rm sym}
&=\min_{\mathcal{S}\subseteq\{1, 2,\ldots,\Nt\}}\frac{\Nt}{|\mathcal{S}|}\log\det\left(\svv{I}+\mathcal{H}_\mathcal{S}^H\mathcal{H}_\mathcal{S}\right)
\end{align}
where 
$\mathcal{H}_\mathcal{S}$ is defined as the submatrix of $\mathcal{H}_\mathcal{S}$ generated from taking only the channel matrices $\svv{H}_i$ corresponding to user indices $i$ such that  $i\in \mathcal{S}$.

In order to leverage the bounds derived for the scalar MAC scenario, we may use the simple bound 
% To derive a bound which holds for this case, we note that for any $\mathcal{H}_\mathcal{S}$ we have 
(see, e.g. \cite{Paulraj2000}, Equation (5))
\begin{align}
    \log\det\left(\svv{I}+\mathcal{H}_\mathcal{S}^H\mathcal{H}_\mathcal{S}\right)\geq \log\det\left(\svv{I}+\|\mathcal{H}_S\|^2_F\right),
    \label{eq:ForbNormSmallerThanDet}
\end{align}
where $\|\svv{A}\|_F$ denotes the Frobenius norm of a matrix $\svv{A}$.
Denote the ``Frobenius-norm mutual information" by 
\begin{align}
    \tilde{C}(\mathcal{S})=\log\det\left(\svv{I}+\|\mathcal{H}_\mathcal{S}\|^2_F\right)
    \label{eq:tildeC_S}
\end{align}
and
\begin{align}
    \tilde{C}=\log\det\left(\svv{I}+\|\mathcal{H}\|^2_F\right)
    \label{eq:tildeC}.
\end{align}
It follows that for any channel realization $C \geq \tilde{C}$ and similarly, for any subset of users $\mathcal{S}$, we have $C(\mathcal{S}) \geq \tilde{C}(\mathcal{S})$.
%We note that when $C$ is small, the gap between the $C$ and $\tilde{C}$ is small. 

%we can extend the achievable bound derived for the case of i.i.d Rayleigh-fading channel with single antenna at each terminal to the general case. 

Considering now the performance of a protocol where all users transmit at a rate that is just below  $\tilde{C}/N$, the counterparts of Lemma~\ref{lem:lem1} and Theorem~\ref{thm:thm2} are the following. 

% hence the probability of achieving $\tilde{C}$ is not far from the probability of achieving $C$. We begin with the following lemma from which Theorem~\ref{thm:thm3} follows. 

\begin{lemma}
For a symmetric $N$-user $N_r \times N_t$  Rayleigh-fading MIMO-MAC with Frobenius-norm mutual information $\tilde{C}$,
for any  subset of users $\mathcal{S}\subseteq\{1, 2,\ldots,\Nt\}$ with cardinality $k$, we have
 \begin{align}
     &\prob\left(\frac{\Nt}{|\mathcal{S}|}\tilde{C}(\mathcal{S})<R_{}|\tilde{C}\right) \nonumber \\
     &= \frac{\mathcal{B}(\frac{2^{R_{}|\mathcal{S}|/N}-1}{2^{\tilde{C}}-1};|\mathcal{S}|N_r N_t,(\Nt-|\mathcal{S}|)N_rNt)}{\mathcal{B}(1;|\mathcal{S}|N_rN_t,(\Nt-|\mathcal{S}|)N_rN_t)}\nonumber \\
     &\triangleq \tilde{P}_{\rm out}(k,R_{}|\tilde{C})
 \end{align}
 where $\mathcal{B}(x;a,b)$ is the incomplete beta function, $\tilde{C}(\mathcal{S})$ is defined in \eqref{eq:tildeC_S} and  $\tilde{C}$ is defined in \eqref{eq:tildeC}.
 \label{lem:lem2}
\end{lemma}
\begin{proof}
Denoting by $\svv{h}_{\rm vec}$ the vectorization of $\mathcal{H}$, we have
\begin{align}
    \tilde{C}=\log\left(1+\sum_{i=1}^{N_rN_tN}|{h}_{\rm vec,i}|^2\right).
\end{align}

As noted in the previous section, conditioned on \mbox{$\tilde{C}$, $\svv{h}_{\rm vec}\triangleq(h_1,\ldots,h_{N_rN_tN})$} is uniformly distributed over
an $N_rN_tN$-dimensional complex sphere of radius $\sqrt{2^ {\tilde{C}}-1}$ and hence $\svv{h}_{\rm vec}/\| \svv{h}_{\rm vec} \|$ may be viewed as the first row of a unitary matrix $\svv{U}$ drawn at random according to the Haar measure.

By symmetry, for any set $\mathcal{S}$ with cardinally $k$, the distribution of
$\tilde{C}(\mathcal{S})$ is equal to that of
\begin{align}
    \tilde{C}(\{1,2,\ldots,k\})
    &=\log\left(1+\sum_{i=1}^{N_rN_tk} |h_{\rm vec,i}|^2\right) \nonumber \\
                  &=\log\left(1+\left(2^{\tilde{C}}-1\right)\sum_{i=1}^{N_rN_tk} |U_{1,i}|^2\right).
  %\label{eq:58}
\end{align}
Denoting the partial sum of $k$ entries
as $\displaystyle{X=\sum_{i=1}^{N_rN_tk} |U_{1,i}|^2}$, we therefore have
\begin{align}
   \prob\left(\frac{\Nt}{k}\tilde{C}(\mathcal{S})<R \Big|\tilde{C}\right)&=\prob\left( 1+ \left(2^{\tilde{C}}-1\right)X <2^{R\frac{k}{\Nt}}\right)
   \nonumber \\
   & = \prob\left(X<\frac{2^{R\frac{k}{\Nt}}-1}{2^{\tilde{C}}-1}\right).
\end{align}
the rest of the proof follows the footsteps of the proof of Lemma~\ref{lem:lem1}.
\end{proof}

\begin{theorem}
For a symmetric $N$-user $N_r \times N_t$  Rayleigh-fading MIMO-MAC, we have 
\begin{align}
    \prob\left(C_{\Sigma-\rm sym}<R|\tilde{C}\right)\leq \sum_{k=1}^{\Nt} {{\Nt}\choose{k}} \tilde{P}_{\rm out}(k,R|\tilde{C})
\end{align}
\label{thm:thm3}
\end{theorem}
\begin{proof}
By \eqref{eq:ForbNormSmallerThanDet} and \eqref{eq:tildeC_S}, for every $\mathcal{S}$ and channel realization, it holds that 
\begin{align}
    \tilde{C}(\mathcal{S})\leq C(\mathcal{S})
\end{align}
Therefore,
\begin{align}
    \prob\left(C_{\Sigma-\rm sym}<R \Big|\tilde{C}\right)\leq\prob\left(\tilde{C}_{\rm sym}<R \Big| \tilde{C}\right),
\end{align}
and similar to Theorem~\ref{thm:thm1}, applying the union bound, we get
\begin{align}
\prob\left(\tilde{C}_{\rm sym}<R \Big|\tilde{C}\right)\leq \sum_{k=1}^{\Nt} {{\Nt}\choose{k}} \tilde{P}_{\rm out}(k,R|\tilde{C}).
\end{align}
\end{proof}

Figure~\ref{fig:empVsFrob_3X2_2users_rTar} depicts a comparison between the empirical outage probability and the upper bound provided by  Theorem~\ref{thm:thm3} for the case of two users, each equipped with $2$  antennas and a receiver equipped with $3$ antennas, where the target rate is set to $3$ bits. The  outage probability was evaluated empirically by  Monte-Carlo simulation. %of Rayleigh-fading channels. 
To calculate the bound, the Frobenius norm of each channel matrix drawn was calculated.
%and was used to calculate the resulting outage probability bound.

It can seen that at high $\SNR$, the slope of the bound is similar to that of the empirical results. Recalling the MIMO-MAC DMT, we note that since the target rate is fixed (is not a function of the $\SNR$), the slope at high $\SNR$ is in fact the maximal diversity offered in this configuration, i.e., the diversity corresponding to zero multiplexing gain. Recalling the DMT of the symmetric capacity \eqref{eq:DMT_symCap}, the latter is $N\cdot N_r \cdot N_t$ which matches the slope given by Theorem~\ref{thm:thm2}.
On the other hand,  relying on the Frobenius norm results in a loose bound  at low  values of SNR (high outage probabilities). 
%both asymptotic  slope of the empirical results as w approaches $N\cdot N_r \cdot N_t$ at high $\SNR$. 
%As was demonstrated above, using the Frobenius norm equals to analyzing the vectorized ($N_r \cdot N_t \times 1$) channel for each user. Therefore, using \eqref{eq:DMT_symCap}, the maximal diversity of the resulting MIMO-MAC is also $N\cdot N_r \cdot N_t$.
 
% This is not the case at low $\SNR$. Since the Frobenius is smaller than or equal to the sum-capacity, at low $\SNR$ there are cases that the Frobenius is smaller than the target rate hence the bound results with outage probability equal $1$ while in practice it can be supported by this channel (albeit, not in high probability resulting with high outage probability). At high $\SNR$ the ``tail'' behaviour of the Frobenius and the sum-capacity is the same hence we get the same slope.

\begin{figure}[htbp]
\begin{center}
\includegraphics[width=\figSize\columnwidth]{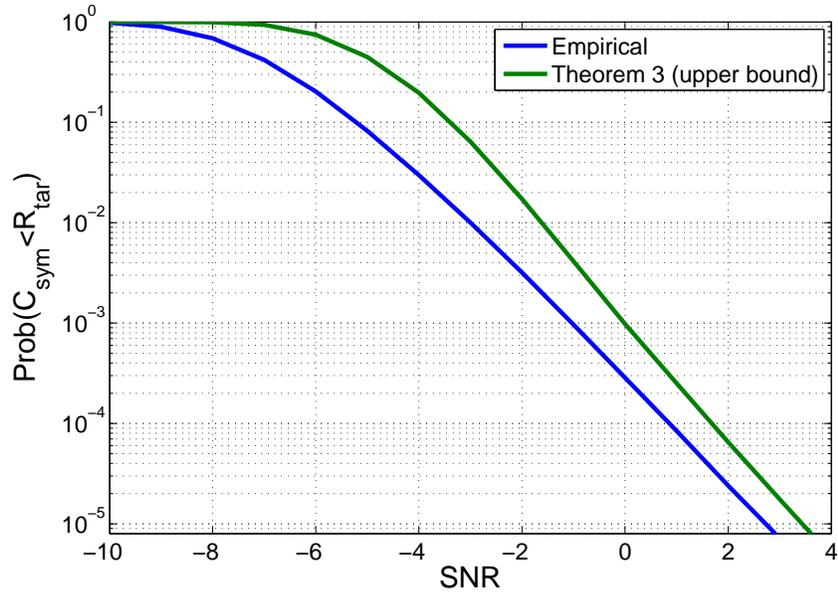}
\end{center}
\caption{Comparison of empirical outage and
the upper bound provided by Theorem~\ref{thm:thm3}
for a two-user $3 \times 2$ i.i.d. Rayleigh-fading MAC. The target rate is set to $3$ bits.}
\label{fig:empVsFrob_3X2_2users_rTar}
\end{figure}

%\begin{remark}
We may obtain a tighter bound for low SNR values, for the special case of $N=2$ users, where each is equipped with a single antenna and the  receiver is equipped with $N_r \geq 2$ antennas. 
Specifically, in \cite{domanovitz2017explicit} a different upper bound for the outage probability was derived in the context of a randomly precoded   compound single-user $N_r\times 2$ MIMO channel. It is easy to verify that the derived bound carries over to the setting considered in the present paper, when rewritten as follows:\footnote{The main step is to recall that the SVD decomposition of an i.i.d. circularly-symmetric complex Gaussian matrix yields left and right singular vector matrices that are uniformly (Haar) distributed, as is the precoding matrix considered in the analysis of \cite{domanovitz2017explicit}.}

\begin{theorem}[Theorem 2 in \cite{domanovitz2017explicit}]
For a two-user i.i.d. Rayleigh-fading MAC where each user is equipped with a single antenna and the receiver is equipped with $N_r$ antennas,  
\begin{align}
    \prob\left(C_{\Sigma-\rm sym}<R|{C}\right) \leq 1-\sqrt{1-2^{-(C-R)}}.
\end{align}
\label{thm:thm4}
\end{theorem}
% We note that Theorem 2 in \cite{domanovitz2017explicit} bounds the outage probability of a CUE-precoded $N_r\times 2$ compound MIMO channel with white-input mutual information $C$ and $N_r\geq 2$ (by finding the worst case combination of the singular values of the equivalent channel). Following remark 3 in \cite{domanovitz2018outage} we further note that since we restrict our attenuation to i.i.d. Rayleigh-fading channels, the random transformation is actually performed
% by nature.
The main advantage Theorem~\ref{thm:thm4} Figure with respect to Theorem~\ref{thm:thm3} is that the conditioning is on $C$, the true sum capacity of the channel,  rather than its Frobenius-norm counterpart. 

Figure~\ref{fig:empVsFrob_6X1_2users_rTar}  depicts the empirical outage probability, the upper bound of Theorem~\ref{thm:thm3} and the bound of Theorem~\ref{thm:thm4} for the case of a symmetric two-user $6 \times 1$ Rayleigh-fading MAC, where the target rate is set to $3$ bits. It can be seen that at low SNR, Theorem~\ref{thm:thm4} provides a tighter bound than Theorem~\ref{thm:thm3} but it becomes loose rapidly as it does not capture the maximal diversity offered by the system.
%hence at high $\SNR$ Theorem~\ref{thm:thm3} is better.

\begin{figure}[htbp]
\begin{center}
\includegraphics[width=\figSize\columnwidth]{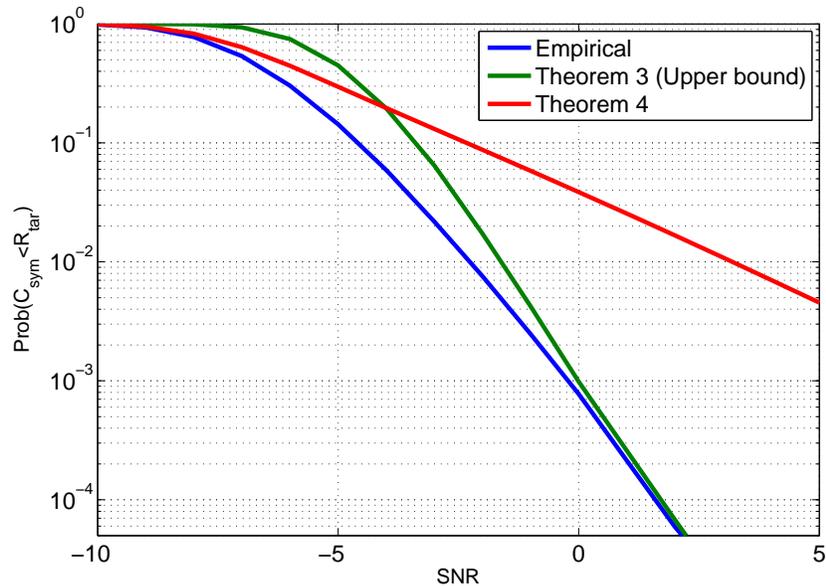}
\end{center}
\caption{Comparison of empirical outage and  the upper bound provided by Theorems~\ref{thm:thm3} and \ref{thm:thm4} for a symmetric two-user $6 \times 1$ i.i.d.  Rayleigh-fading MAC.
%, each user with $1$ transmit antennas and a receiver with $6$ antennas, with 
The target rate is set to $3$ bits.}
\label{fig:empVsFrob_6X1_2users_rTar}
\end{figure}

%\end{remark}

\section{Practical Realization of the Communication Protocol via Precoded Integer Forcing}
\label{sec:IF}

In this section we  empirically demonstrate the effectiveness of the integer-forcing (IF) receiver when used in conjunction with unitary space-time precoding as a practical transmission scheme in the context of  the considered communication protocol. Due to space limitations, we refer the
reader to \cite{IntegerForcing} for a description of the integer forcing framework and its implementation.

%As mentioned in Section~\ref{sec:lowerbound}, $C_{\Sigma-\rm sym}$ serves as an %upper bound on the achievable rate of an IF receiver.
When it comes to fading channels,
it has been shown in \cite{IntegerForcing} that the IF receiver achieves the DMT over i.i.d. Rayleigh-fading channels where the number of receive antennas is greater or equal to the number of transmit antennas.

We observe that this does not hold in the general case; in particular, IF does not achieve the DMT for the case of a MAC where all terminals are equipped with a single antenna.
Specifically, Figure~\ref{fig:logErr_ML_IF_T2} depicts (in logarithmic scale) the empirical outage probability of the IF  receiver and the exact outage probability for optimal communication (Gaussian codebooks and  ML decoding), as given by Theorem~\ref{thm:thm1}, for the two-user i.i.d. Rayleigh-fading MAC. The symmetric rate achieved by a given scheme is denoted by $R_{\rm scheme}$.

It is evident that the slopes are different. This raises the question of whether IF is inherently ill-suited for the MAC channel. A negative answer to this question may be inferred by recalling some further lessons from the DMT analysis of the MAC.

While the optimal DMT for the i.i.d. Rayleigh-fading MAC was
derived in \cite{2004:DMT_MAC} using Gaussian codebooks of
sufficient length, it was subsequently shown that the  DMT of the MAC can be achieved using structured codebooks by combining uncoded QAM constellations with
% analyzed for the case in which the block length $l$ is large enough, i.e., in our case $l\geq\Nt$. This suggests that IF performance can be improved when combined with space-time precoding.  It is worth noting that each transmitter is applying the precoding independently hence the distributed nature of the problem is not violated.
space-time unitary precoding (and ML decoding). Specifically, such a  MAC-DMT achieving  construction is given in  \cite{lu2011dmt}.
This raises the possibility that the sub-optimality of the IF receiver observed in
 Figure~\ref{fig:logErr_ML_IF_T2} may at least in part be remedied by applying unitary space-time precoding at each of the transmitters.
 We note  that each transmitter applies precoding only to its own data streams so the distributed nature of the problem is not violated.

Following this approach, we have implemented the IF receiver with  unitary space-time precoding applied at each transmitter. We have employed random (Haar) precoding (with independent matrices drawn for the different users) over two ($T=2$) time instances as well as deterministic precoding using the
matrices proposed in
%As a reference we present the outage achieved when using the code
\cite{badr2008distributed}.\footnote{When using an ML receiver, this space-time code is known to achieve the DMT for multiplexing rates $r\leq\frac{1}{5}$. As detailed in \cite{lu2011dmt}, whether this code achieves the optimal MAC-DMT also when $r> \frac{1}{5}$ remains an open question.}

These matrices can be expressed as
\begin{align}
    \svv{P}_{st,c}^1&=\frac{1}{\sqrt{5}}
    \begin{bmatrix}
    \alpha & \alpha\phi \\
    \bar{\alpha} & \bar{\alpha}\bar{\phi}
    \end{bmatrix},~&\svv{P}_{st,c}^2=\frac{1}{\sqrt{5}}
    \begin{bmatrix}
    j\alpha & j\alpha\phi \\
    \bar{\alpha} & \bar{\alpha}\bar{\phi}
    \end{bmatrix}
    \label{eq:Badr_Bel}
\end{align}
where
\begin{align}
    \phi&=\frac{1+\sqrt{5}}{2},&\bar{\phi}&=\frac{1-\sqrt{5}}{2} \nonumber \\
    \alpha&=1+j-j\phi,&\bar{\alpha}&=1+j-j\bar{\phi}.
\end{align}
%The results appear
%We also compare random space-time precoding over two channel uses with IF.

We also replot Figure~\ref{fig:logErr_ML_IF_T2} in terms of PDF (rather than CDF) as Figure~\ref{fig:logErr_ML_IF_T2,PDF}, but without random Haar space-time precoding (so as to avoid ``clutter"). As can be seen, the precoding matrices in \eqref{eq:Badr_Bel} improve the outage probability for most target rates.
\newline
%We also simulated the performance of random (Haar measure) unitary space-time precoding and such precoding also
%As using random we do not plot it in the figure (so as to avoid ``clutter").

% We further note that in addition to standard IF, we also implemented a variant that incorporates successive interference cancellation, referred to as IF-SIC \cite{OrdentlichErezNazer:2013}. As can be seen, IF-SIC   results in a significant improvement for all precoding schemes used.

% As can be seen from Figure~\ref{fig:logErr_ML_IF_T2}, both random space-time precoding and the precoding matrices in \eqref{eq:Badr_Bel} improve the outage probability for most target rates.

\begin{figure}[htbp]
\begin{center}
\includegraphics[width=\figSize\columnwidth]{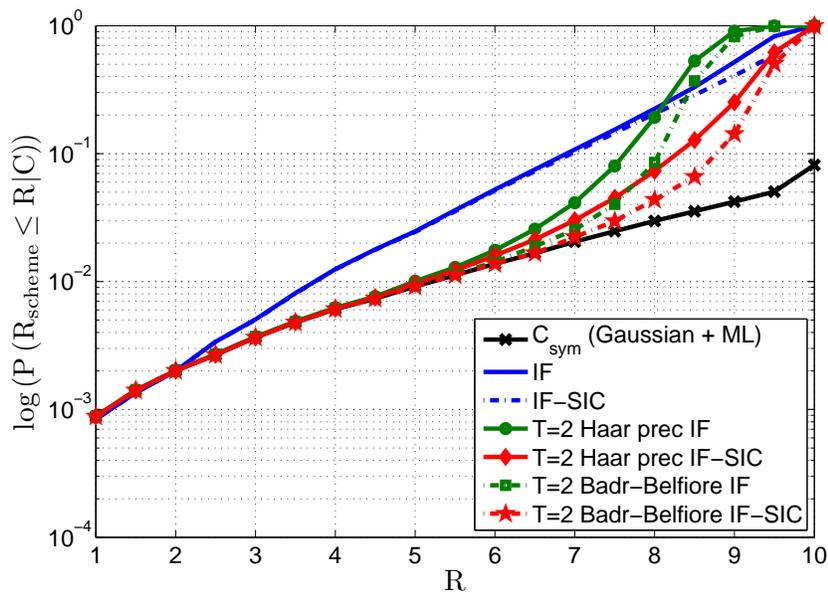}
\end{center}
\caption{Outage probability of  linear codes (with and without space-time precoding) with IF equalization versus  Gaussian codebooks with ML decoding for a two-user i.i.d. Rayleigh-fading  MAC with sum capacity $C_{}=10$.}
%\caption{Probability distribution function  of the rate achieved with  linear codes (with and without space-time precoding) in conjunction with IF equalization versus  that achieved Gaussian codebooks with ML decoding for a two-user i.i.d. Rayleigh fading  MAC with sum capacity $C_{}=10$.}
\label{fig:logErr_ML_IF_T2}
\end{figure}

\begin{figure}[htbp]
\begin{center}
\includegraphics[width=\figSize\columnwidth]{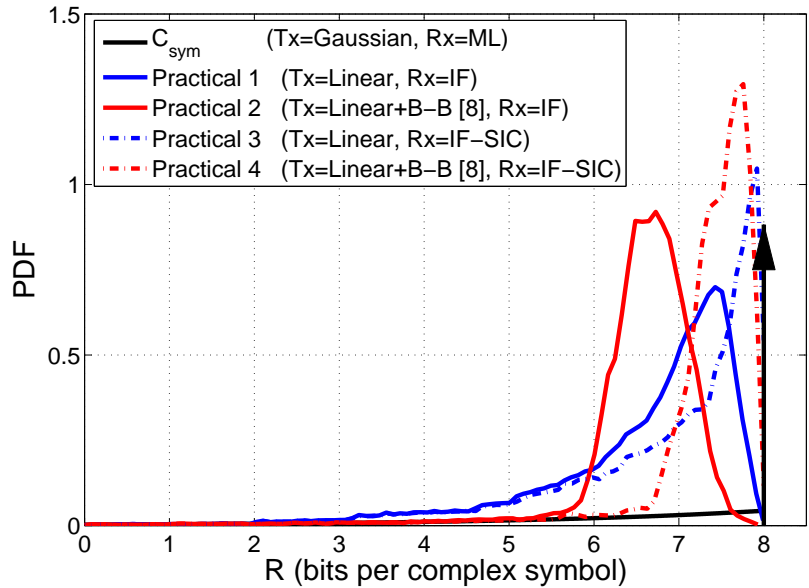}
\end{center}
%\caption{Outage probability of  linear codes (with and without space-time precoding) with IF equalization versus  Gaussian codebooks with ML decoding for a two-user i.i.d. Rayleigh fading  MAC with sum capacity $C_{}=10$.}
\caption{Probability distribution function  of the rate achieved with  linear codes (with and without space-time precoding) in conjunction with IF equalization versus  that achieved Gaussian codebooks with ML decoding for a two-user i.i.d. Rayleigh-fading  MAC with sum capacity $C_{}=10$.}

\label{fig:logErr_ML_IF_T2,PDF}
\end{figure}

We further note that in addition to standard IF, we also implemented a variant that incorporates successive interference cancellation, referred to as IF-SIC \cite{OrdentlichErezNazer:2013}. As can be seen, IF-SIC results in a significant improvement for all precoding schemes used.

In Figure~\ref{fig:Rout_Csum_ML_IF_T2} we study the average  symmetric rate achieved by different schemes w.r.t. a two-user i.i.d. Rayleigh-fading channel when
we condition on the sum capacity of the channel. We plot the fraction of the sum capacity  attained by the various schemes.
% rate achieved from the sum capacity when all users transmit at the symmetric rate capacity (per user) for the two-user i.i.d. Rayleigh fading MAC versus the fraction achieved when using linear codes (at the maximal achievable rate) in conjunction with IF with the two  different precoding methods as described above.
%The symmetric capacity is compared against random space-only precoding, %random space-time precoding (which utilizes two Channel uses) and the Badr-Belfiore precoding described in \eqref{eq:Badr_Bel}.
We first observe that IF-SIC combined with space-time precoded linear codes achieves a large fraction of the symmetric capacity. Further, as can be seen, the fraction of the sum capacity achieved by all the different schemes considered approaches one as the sum capacity grows.
%ergodic capacity which is close to the symmetric capacity. We also note %that using random space-time precoding attains most of the benefit gained %by the Badr-Belfire precoding.

\begin{figure}[htbp]
\begin{center}
\includegraphics[width=\figSize\columnwidth]{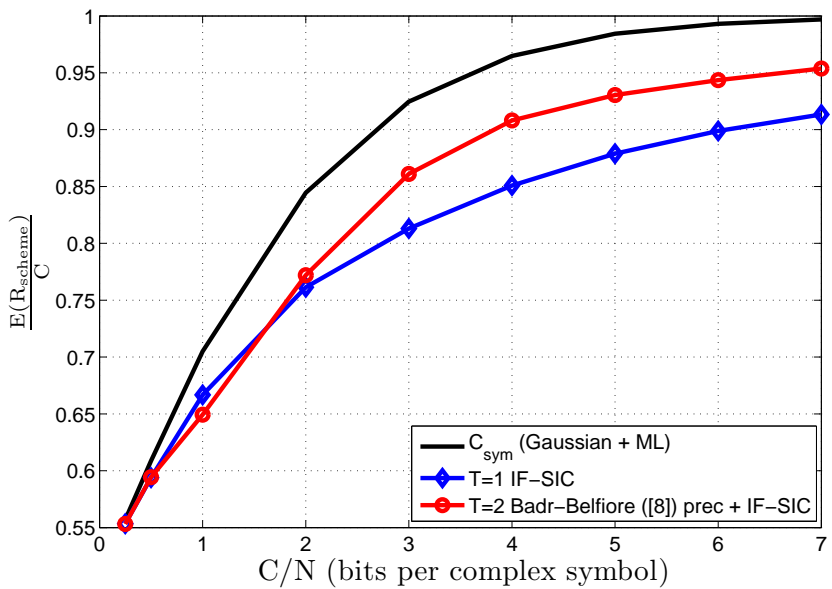}
\end{center}
\caption{Average rate conditioned on the sum capacity when using linear codes (with and without space-time precoding) with IF equalization versus Gaussian codebooks with ML decoding, over a  two-user normalized (conditioned)  i.i.d. Rayleigh-fading MAC.}
\label{fig:Rout_Csum_ML_IF_T2}
\end{figure}

Finally, in Figure~\ref{fig:outageCapacity2_4_6}, we plot the fraction of the sum capacity that is achieved, allowing for a fixed outage probability, by the proposed protocol where  we consider both the ideal performance achieved as captured by the symmetric capacity and the rate achieved using IF in conjunction with SIC. As can be observed, the performance of IF-SIC for small outage probabilities is very close to the theoretical limits of the considered transmission protocol. We note, however, that as the number of users increases (and also, as the sum-capacity increases), the problem of finding a ``good'' integer matrix (as required in IF equalization) becomes computationally difficult and may result in compromised rates when using practical sub-optimal algorithms such as the LLL algorithm  to find candidate integer matrices. 
%This may explain the 
%Therefore, for higher dimensions (or equivalently %when combining the IF with space-time which results %with searching for an integer matrix in dimension %5which is quadratic to the physical channel) we got %much worse results. We believe that more %sophisticated methods for lattice basis reduction would result with better performance, thus there is a trade off between performance and complexity.

\begin{figure}[htbp]
\begin{center}
\includegraphics[width=\figSize\columnwidth]{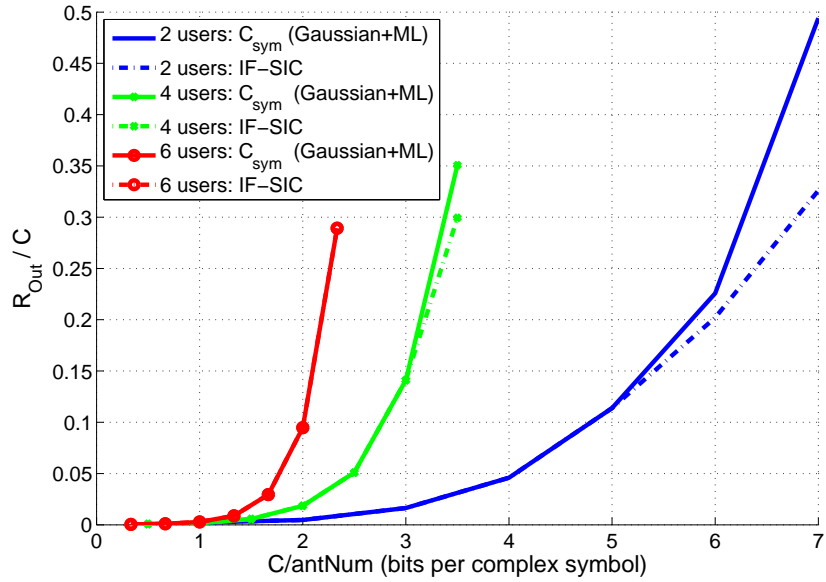}
\end{center}
\caption{Fraction of the sum capacity achieved at  1\% outage probability by the proposed transmission protocol over a scalar i.i.d. Rayleigh-fading MAC.  The performance limits, as captured by by $C_{\rm sym}$  is depicted as a function of the sum capacity normalized by the number of users, for $N=2,4,6$ users. The performance limits of IF-SIC equalization are also depicted.}
\label{fig:outageCapacity2_4_6}
\end{figure}
\section{Conclusions}
\label{sec:conc}
We analyzed the performance of a simple communication protocol for transmission over the scalar Rayleigh-fading MAC, where  all users transmit at just below the symmetric capacity (normalized per user) of the channel. 
Tight bounds were established on the distribution of the achievable rate of the protocol. The derived bounds may be  viewed as a significant tightening of the diversity multiplexing tradeoff analysis of the channel. It was further demonstrated that integer-forcing equalization in conjunction with ``space-time'' precoding (over the time axis only) offers a practical means to approach the theoretical limits of the proposed protocol. 

%\clearpage
\bibliographystyle{IEEEtran}
\bibliography{eladd}

\end{document}